\documentclass[letterpaper, 10 pt, conference]{ieeeconf}
\IEEEoverridecommandlockouts 
\usepackage{amsmath}	
\usepackage{amsfonts}
\usepackage{subcaption}
\usepackage{amssymb}
\usepackage{hyperref}
\usepackage{algorithm, algpseudocode}  
\usepackage{array}                      
\usepackage{siunitx}
\usepackage{cite}
\usepackage{wrapfig}
\usepackage{lipsum}
\newcolumntype{C}{>$c<$}
\usepackage{graphicx}
\usepackage{mathtools}
\usepackage{color,soul}
\providecommand{\U}[1]{\protect\rule{.1in}{.1in}}                                                                                           
\newtheorem{theorem}{Theorem}  
\newtheorem{assumption}{Assumption}

\newtheorem{definition}{Definition}

\newtheorem{problem}{Problem}

\newtheorem{remark}{Remark}

\newcommand{\norm}[1]{\left\lVert#1\right\rVert}

\let\oldIEEEkeywords\IEEEkeywords
\def\IEEEkeywords{\oldIEEEkeywords\normalfont\bfseries\ignorespaces}
\newcommand{\argmin}{\mathrm{arg}\min}

\begin{document}
	
	\title{Controller Synthesis for Multi-Agent Systems With Intermittent Communication: A Metric Temporal Logic Approach}
	
	\author{Zhe~Xu, Federico M. Zegers, Bo Wu, Warren Dixon and Ufuk Topcu
		\thanks{Zhe~Xu and Bo Wu are with the Oden Institute
			for Computational Engineering and Sciences, University of Texas,
			Austin, Austin, TX 78712, Federico M. Zegers and Warren Dixon are with the Department of Mechanical
            and Aerospace Engineering, University of Florida, Gainesville, Florida 32611, Ufuk Topcu is with the Department
			of Aerospace Engineering and Engineering Mechanics, and the Oden Institute
			for Computational Engineering and Sciences, University of Texas,
			Austin, Austin, TX 78712, e-mail: zhexu@utexas.edu, fredzeg@ufl.edu, bwu3@utexas.edu, wdixon@ufl.edu, utopcu@utexas.edu.}     
	}

	\maketitle
	
\begin{abstract} 
This paper develops a controller synthesis approach for a multi-agent system (MAS) with
intermittent communication. We adopt a \textit{leader-follower} scheme, where a mobile \textit{leader} with absolute position sensors switches among
a set of \textit{followers} without absolute position sensors to
provide each follower with intermittent state information. We model the MAS as a switched system. The followers are to asymptotically reach a predetermined \textit{consensus} state. To guarantee the stability of the switched system and the consensus of the followers, we derive \textit{maximum and minimal dwell-time conditions} to constrain the intervals between consecutive time instants at which the leader should provide state information to the same follower. Furthermore, the leader needs to satisfy practical constraints such as charging its battery and staying in specific regions of interest. Both the maximum and minimum dwell-time conditions and these practical constraints can be expressed by \textit{metric temporal logic} (MTL) specifications. We iteratively compute the optimal control inputs such that the leader satisfies the MTL specifications, while guaranteeing stability and consensus of the followers. We implement the proposed method on a case study with three mobile robots as the followers and one quadrotor as the leader.
\end{abstract}     
	
	\section{Introduction}
	\label{sec_intro}
    
	Coordination strategies for multi-agent systems (MAS)
	have been traditionally designed under the assumption that state feedback is continuously
	available. However, continuous communication over
	a network is often impractical, especially in mobile robot applications
	where shadowing and fading in the wireless communication can cause unreliability, and each agent has limited energy resources \cite{goldsmith2005wireless, Bo_ADHS}. 
	
	Due to these
	constraints, there is a strong interest in developing
	MAS coordination methods that rely on intermittent information over a communication network. The results in \cite{Wang.Lemmon2009,Meng.Chen2013,Cheng.Kan.ea2017,Li.Liao.ea2015,Heemels.Donkers2013,tabuada2007event}
	develop \textit{event-triggered} and \textit{self-triggered} controllers that utilize sampled
	data from networked agents only when triggered by conditions that ensure
	desired stability and performance properties. However, these results
	require a network represented by a strongly connected graph to enable agent coordination.
   This requirement of a \textit{strongly connected network} induces 
   constraints on the motion of the individual agents and additional maneuvers that may deviate
   from their primary purpose. Event-triggered and self-triggered control methods
   can also be used to coordinate the agents that communicate with a central base station or \textit{cloud}
    intermittently as in \cite{Nowzari_Pappas2016}, where submarines
   intermittently surface to obtain state information about themselves and
   their neighbors from a cloud. However, such a coordination strategy
   also requires additional maneuvers from the submarines that detract
   from their primary purpose. 
	
 Depending on the application and/or environment, some of the agents in
 a MAS may not be equipped with absolute position sensors. In such
 scenarios, the results in \cite{Wang.Lemmon2009,Meng.Chen2013,Cheng.Kan.ea2017,Li.Liao.ea2015,Heemels.Donkers2013,tabuada2007event}
are invalid. Therefore, there is a need for distributed methods capable
of coordinating these agents that are not equipped with
absolute position sensors while utilizing intermittent information.
Moreover, such methods should not require agents to perform additional
maneuvers to ensure the connectivity of the network. In \cite{Zegers.Chen.ea2019}, a
set of \textit{followers} operating with inaccurate position sensors 
are able to reach consensus at a desired state while a \textit{leader} intermittently         
provides each follower with state information. By introducing a
leader, the followers are able to perform their tasks without
the need to perform additional maneuvers to obtain state information.
	
Building on the work of \cite{Zegers.Chen.ea2019}, we adopt a \textit{leader-follower} scheme, where the MAS is modeled as a switched system \cite{Wu2019SwitchedLS,Bo2018CSL}. As an illustrative example shown in Fig. \ref{fig_intro}, the three followers need to reach consensus at the center of the green feedback region and one leader agent is to provide intermittent state information to each follower. To guarantee the stability of the switched system and consensus of the followers, we derive \textit{maximum and minimal dwell-time conditions} to constrain the intervals between consecutive time instants at which the leader should provide state information to the same follower. 

The maximum and minimum dwell-time conditions can be encoded by \textit{metric temporal logic} (MTL) specifications \cite{Ouaknine2005}. Such specifications have also been used in many robotic applications for time-related specifications \cite{zhe_ijcai2019}. Furthermore, as the leader is typically more energy-consuming and safety-critical due to the high-quality sensing, communication and mobility equipments, the leader is likely required to satisfy additional MTL specifications for practical constraints such as charging its battery and staying in specific regions. 
In the example shown in Fig. \ref{fig_intro}, the leader needs to satisfy an MTL specification ``\textit{reach the charging station $G_1$ or $G_2$ in every 6 time units and always stay in the yellow region $D$}''. 
	
We design the followers' controllers such that guarantees on the stability of the switched system and consensus of the followers hold, provided that the maximum and minimal dwell-time conditions are satisfied. Then we synthesize the leader's controller to satisfy the same MTL specifications that encode the maximum and minimal dwell-time conditions and the additional practical constraints. There is a rich literature on controller synthesis subject to temporal logic specifications \cite{KHFP,Nok2012,BluSTL,sayan2016,Zhiyu2017ACC,Zhiyu2017CDC,Bo2019,zhe_advisory,zheACC, zhe_control,zheACC2018}. For linear or switched linear systems, the controller synthesis problem can be converted into a mixed-integer linear programming (MILP) problem~\cite{BluSTL,sayan2016}. Additionally, as the followers are not equipped with absolute position sensors, we design an observer to estimate the followers' states and the state estimates can change abruptly due to the intermittent communication of state information. Therefore, we solve the MILP problem iteratively to account for such abrupt changes.
	
We provide an implementation of the proposed method on a simulation case study with three mobile robots as the followers and one quadrotor as the leader. The results in two different scenarios show that the synthesized controller can lead to satisfaction of the MTL specifications, while achieving the stability of the switched system and consensus of the followers.

	\begin{figure}
		\centering
		\includegraphics[scale=0.3]{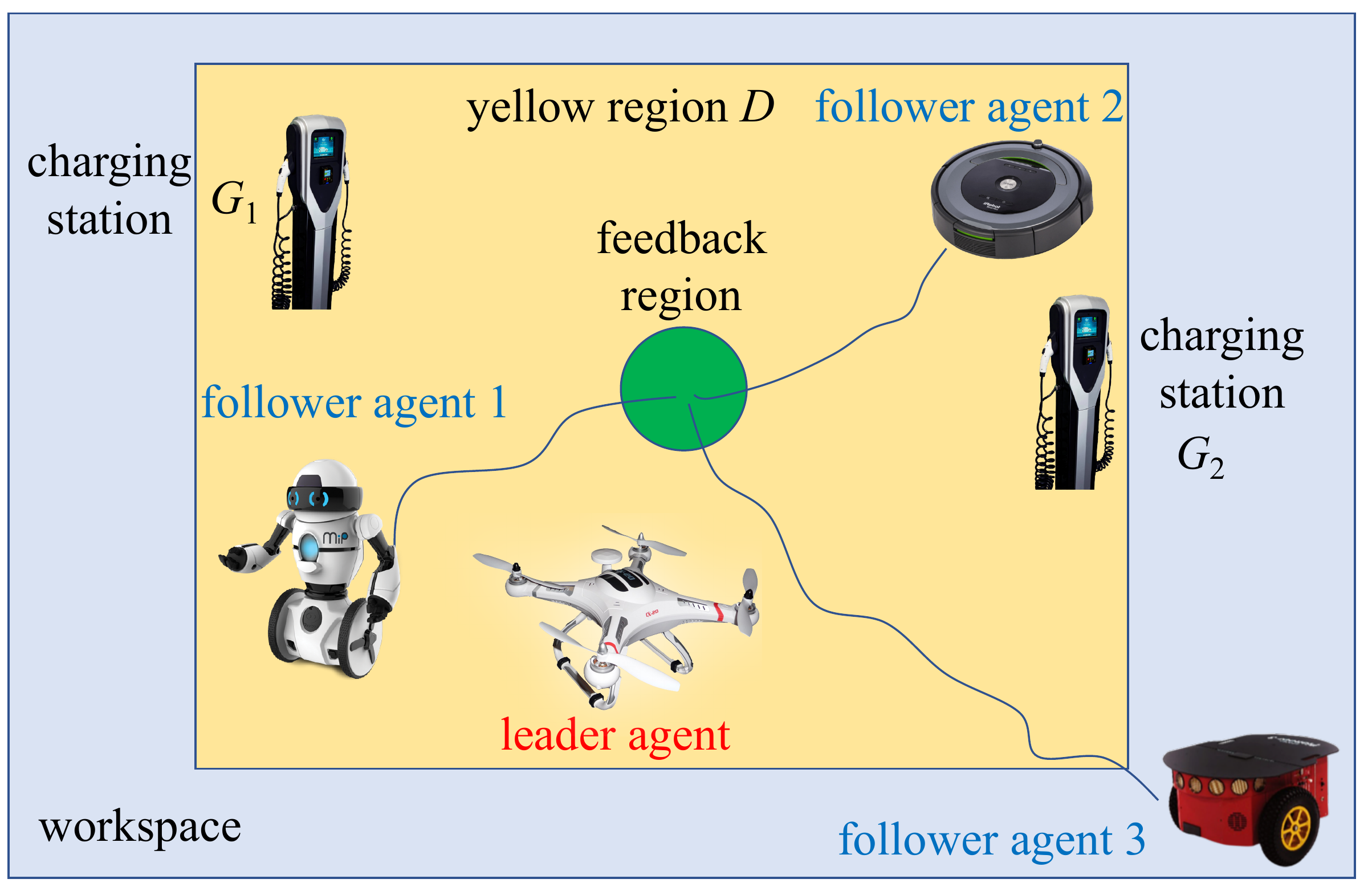}
		\caption{Illustrative example of a MAS with a leader (quadrotor) and three followers (mobile robots).}  
		\label{fig_intro}
	\end{figure}

\section{Background and Problem Formulation}

\subsection{Agent Dynamics}
\label{sec_dynamics}

Consider a multi-agent system (MAS) consisting of $Q$ followers ($Q\in\mathbb{Z}_{>0}$\footnote{$\mathbb{Z}_{>0}$ denotes the
set of positive integers.}) index by $i\in F\triangleq\left\{ 1,...,Q\right\}$
and a leader indexed by $0$. Let the time set be $\mathbb{T} = \mathbb{R}_{\ge0}$.
Let $y_{0},\text{ }y_{i}: \mathbb{T}\rightarrow\mathbb{R}^{z}$
denote the position of the leader and follower $i$, respectively. Let $x_{0}: \mathbb{T}\rightarrow\mathbb{R}^{l}$ and $x_{i}: \mathbb{T}\rightarrow\mathbb{R}^{m}$
denote the state of the leader and follower $i$, respectively.
The linear time-invariant dynamics of the leader and follower $i$
are 
\begin{align}
\begin{split}
\dot{x}_{0}\left(t\right) & =  A_0x_{0}\left(t\right)+B_0u_{0}\left(t\right), \\
 y_{0}\left(t\right) & =  C_0x_{0}\left(t\right), \\
\dot{x}_{i}\left(t\right) & =  Ax_{i}\left(t\right)+Bu_{i}\left(t\right)+d_{i}\left(t\right), \\
 y_{i}\left(t\right) & = Cx_{i}\left(t\right),
 \end{split}
\label{Follower Dynamics}
\end{align}
where $A_0\in\mathbb{R}^{l\times l}, A\in\mathbb{R}^{m\times m}$, $B_0\in\mathbb{R}^{l\times n}, B\in\mathbb{R}^{m\times n}$, $C_0\in\mathbb{R}^{z\times l}, C\in\mathbb{R}^{z\times m}$. Here, $u_{0},\text{ }u_{i}: \mathbb{T}\rightarrow\mathbb{R}^{n}$
denote the control inputs of the leader and follower $i$, respectively,
and $d_{i}: \mathbb{T}\rightarrow\mathbb{R}^{m}$ is an
exogenous disturbance. For simplicity, we assume that $\lambda_{\textrm{max}}\left(A\right)\footnote{$\lambda_{\textrm{max}}\left(A\right)$ denotes the maximum
	singular value of $A$.}\in\mathbb{R}_{>0}$ and $B$ has full row rank. 

\subsection{Sensing and Communication}
\label{sec_sensing}
Each follower is equipped with a relative position sensor and hardware to enable communication with the leader.               
Since the followers lack absolute position sensors, they are
not able to localize themselves within the global coordinate system.
Nevertheless, the followers can use their relative position sensors
to enable self-localization relative to their initially known locations. However,
relative position sensors like encoders and inertial measurement
units (IMUs) can produce unreliable position information since e.g., wheels of mobile robots
may slip and IMUs may generate noisy data. Hence, the $d_{i}\left(t\right)$
term in (\ref{Follower Dynamics}) models the inaccurate position
measurements from the relative position sensor of follower $i$ as
well as any external influences from the environment.
Navigation through the use of a relative position sensor results in
dead-reckoning, which becomes increasingly more inaccurate with time
if not corrected. On the other hand, the leader is equipped with an absolute
position sensor and hardware to enable communication
with each follower. Unlike a relative position sensor, an absolute
position sensor allows localization of the agents within the global coordinate
system.

\textcolor{black}{The followers' task is to reach consensus to a predetermined state $x_{\textrm{g}}\in\mathbb{R}^{m}$. A feedback region (see Fig. \ref{fig_intro}) centered at the position $Cx_{\textrm{g}}\in\mathbb{R}^{z}$ with radius $R_{\textrm{g}}\in\mathbb{R}_{>0}$ is capable of providing state information to each follower $i\in F$
once $\left\Vert y_{i}\left(t\right)-Cx_{\textrm{g}}\right\Vert=\left\Vert Cx_{i}\left(t\right)-Cx_{\textrm{g}}\right\Vert \leq R_{\textrm{g}}$.}
The leader's task 
is to provide state information to each follower while they navigate to $x_{\textrm{g}}$ with the intermittent state information. Both the leader and the followers are equipped with digital communication hardware where communication is only possible at discrete time instants. Let $R_{\textrm{c}}\in\mathbb{R}_{>0}$ and $R_{\textrm{s}}\in\mathbb{R}_{>0}$
denote the communication and sensing radii of each agent, respectively. For simplicity, let $R_{\textrm{c}}=R_{\textrm{s}}\triangleq R$. 

The leader provides state information to the follower $i$ (i.e., \textit{services} the follower $i$) if and only if $\left\Vert y_{i}\left(t\right)-y_{0}\left(t\right)\right\Vert \leq R$ and the communication channel of the follower $i$ is on. 
We define the \textit{communication switching signal} $\zeta_i$ for follower $i$ as $\zeta_i=1$ if the communication channel is on for follower $i$; and $\zeta_i=0$ if the communication channel is off for follower $i$. We use $t_{s}^{i}\ge0$ to indicate the $s^{th}$
servicing time instance for follower $i$. 
Hence, the $(s+1)^{th}$ servicing time instant for follower $i$ is\footnote{For $s=0,$ $t_{0}^{i}$ is the initial time, for simplicity we take $t_{0}^{i}=0.$}  
\begin{equation}\nonumber
t_{s+1}^{i}\triangleq \inf\left\{ t\geq t_{s}^{i}: (\left\Vert y_{i}\left(t\right)-y_{0}\left(t\right)\right\Vert \leq R) \land (\zeta_i(t)=1)\right\} 
\end{equation}
where $\land$ denotes the \textit{conjunction} logical connective.

\subsection{State Observer and Error Dynamics}
The followers, not equipped with absolute position sensors, implement the following model-based
observer to estimate the state of each follower $i\in F$:
\begin{align}
\begin{split}
\dot{\hat{x}}_{i}\left(t\right) & \triangleq  A\hat{x}_{i}\left(t\right)+Bu_{i}\left(t\right), \text{ }t\in\left[t_{s}^{i},t_{s+1}^{i}\right),\\
\hat{x}_{i}\left(t_{s}^{i}\right) & \triangleq x_{i}\left(t_{s}^{i}\right), 
\end{split}
\label{Reset}
\end{align}
where $\hat{x}_{i}: \mathbb{T}\rightarrow\mathbb{R}^{m}$ 
denotes the estimate of $x_{i}$. 

Then we can obtain the position estimate of follower $i$ as 
\begin{align}
\begin{split}
\hat{y}_{i}\left(t\right) & \triangleq C\hat{x}_{i}\left(t\right).
\end{split}
\label{Reset_y}
\end{align}

To facilitate the analysis, we define the following two error signals
\begin{equation}
e_{1,i}\left(t\right)\triangleq\hat{x}_{i}\left(t\right)-x_{i}\left(t\right)
\label{e1i}
\end{equation}
and
\begin{equation}
e_{2,i}\left(t\right)\triangleq x_{\textrm{g}}-\hat{x}_{i}\left(t\right).  
\label{e2i}
\end{equation}
Similar to \cite{Zegers.Chen.ea2019}, we adopt the following assumptions.
\begin{assumption}
	\label{Assumption 1} The state estimate $\hat{x}_{i}$ is initialized
	as $\hat{x}_{i}\left(0\right)=x_{i}\left(0\right)$ for all $i\in F$.
\end{assumption}

\begin{assumption}
	\label{Assumption 2}  The leader has full knowledge of its own state
	$x_{0}\left(t\right)$ for all $t\geq0$ and the initial state $x_{i}\left(0\right)$
    for all $i\in F$.
\end{assumption}

\begin{assumption}
	\label{Assumption 3} The disturbance $d_{i}$ is bounded,
	i.e., $\left\Vert d_{i}\left(t\right)\right\Vert \leq\overline{d}_{i}$
	for all $t\geq0$, where $\overline{d}_{i}\in\mathbb{R}_{>0}$ is a
	known constant. 
\end{assumption}

The control of follower $i$ is as follows:
\begin{equation}
u_{i}\left(t\right)\triangleq-B^{+}A\hat{x}_{i}\left(t\right)+k_{i}B^{+}e_{2,i}\left(t\right) 
\label{Follower Control}
\end{equation}
such that $B^{+}$ denotes the pseudo-inverse of $B$ and $k_{i}\in\mathbb{R}_{>0}$
is a user-defined parameter. Since $B$ has full row rank (see Section \ref{sec_dynamics}), $BB^{+}=I_{m\times m}$, where $I_{m\times m}$ is the identity matrix.

At each servicing time instant $t_{s}^{i}$, with the feedback provided by the leader, the state
estimate $\hat{x}_{i}$ of follower $i$ immediately resets to $x_{i}$. Therefore, the state
estimates follow the dynamics of switched systems\cite{Xuping}.

Substituting (\ref{Follower Dynamics})     
and (\ref{Reset}) into the time-derivative of (\ref{e1i}) yields
\begin{eqnarray}
\dot{e}_{1,i}\left(t\right) & = & Ae_{1,i}\left(t\right)-d_{i}\left(t\right),\text{ }t\in\left[t_{s}^{i},t_{s+1}^{i}\right), \label{e1i Dot closed-loop}\\
e_{1,i}\left(t_{s}^{i}\right) & = & 0_{m}, \label{e1i reset}
\end{eqnarray}
where $0_{m}\in\mathbb{R}^{m}$ is the zero column vector. Substituting (\ref{Reset}) into the time-derivative of (\ref{e2i}) yields

\begin{alignat}{1}
\dot{e}_{2,i}\left(t\right) & =-k_{i}e_{2,i}\left(t\right),\text{ }t\in\left[t_{s}^{i},t_{s+1}^{i}\right), \text{ }\label{e2i Dot closed-loop}\\
e_{2,i}\left(t_{s}^{i}\right) & =x_{\textrm{g}}-x_{i}\left(t_{s}^{i}\right).     \label{e2i Reset}
\end{alignat}

\subsection{Metric Temporal Logic (MTL)}
\label{MTL}   
To achieve the stability of the swicthed system and consensus of the followers while satisfying the practical constraints
of the leader, the requirements of the MAS can be specified in MTL specifications (see details in Section \ref{sec_reactive}). In this subsection, we briefly review the MTL interpreted
over discrete-time trajectories~\cite{FainekosMTL}. The domain of the position of the agents $y$ is denoted by $\mathcal{Y}\subset\mathbb{R}^z$. The domain $\mathbb{B} = \{true, false\}$ is the Boolean domain, and the time index set is $\mathbb{I} = \{0,1,\dots\}$. With slight abuse of notation, we use $y$ to denote the discrete-time trajectory as a function from $\mathbb{I}$ to $\mathcal{Y}$. A set $AP$ is a set of atomic propositions, each mapping $\mathcal{Y}$ to $\mathbb{B}$. The syntax of MTL is defined recursively as follows:
\[
\phi:=\top\mid \pi\mid\lnot\phi\mid\phi_{1}\wedge\phi_{2}\mid\phi_{1}\vee
\phi_{2}\mid\phi_{1}\mathcal{U}_{\mathcal{I}}\phi_{2}
\]
where $\top$ stands for the Boolean constant True, $\pi\in AP$ is an atomic 
proposition, $\lnot$ (negation), $\wedge$ (conjunction), $\vee$ (disjunction) 
are standard Boolean connectives, $\mathcal{U}$ is a temporal operator
representing \textquotedblleft until\textquotedblright~and $\mathcal{I}$ is a time interval of
the form $\mathcal{I}=[j_{1},j_{2}]$ ($j_1\le j_2$, $j_1, j_2\in\mathbb{I}$). We
can also derive two useful temporal operators from \textquotedblleft
until\textquotedblright~($\mathcal{U}$), which are \textquotedblleft
eventually\textquotedblright~$\Diamond_{\mathcal{I}}\phi=\top\mathcal{U}_{\mathcal{I}}\phi$ and
\textquotedblleft always\textquotedblright~$\Box_{\mathcal{I}}\phi=\lnot\Diamond_{\mathcal{I}}\lnot\phi$. We define the set of states that satisfy the atomic proposition $\pi$ as $\mathcal{O}(\pi)\in \mathcal{Y}$. 

Next, we introduce the Boolean semantics of MTL for trajectories of finite length in the strong and the weak view, which are modified from the literature of temporal logic model checking and monitoring \cite{Eisner2003,KupfermanVardi2001,Ho2014}.  We use $t[j]\in\mathbb{T}$ to denote the time instant at time index $j\in\mathbb{I}$ and $y^j\triangleq y(t[j])$ to denote the value of $y$ at time $t[j]$. In the following, $(y^{0:H},j)\models_{\rm{S}}\phi$ (resp. $(y^{0:H},j)\models_{\rm{W}}\phi$)
means the trajectory $y^{0:H}\triangleq y^0\dots y^H$ $(H\in\mathbb{Z}_{\ge0})$ strongly (resp. weakly) satisfies $\phi$ at time index $j$, $(y^{0:H},j)\not\models_{\rm{S}}\phi$ (resp. $(y^{0:H},j)\not\models_{\rm{W}}\phi$)
means $y^{0:H}$ fails to strongly (resp. weakly) satisfy $\phi$ at time index $j$. 

\begin{definition}
	The Boolean semantics of MTL for trajectories of finite length in the strong view is defined recursively as follows~\cite{zhe_advisory}:
	\[
	\begin{split}
	(y^{0:H},j)\models_{\rm{S}}\pi~\mbox{iff}~& j\le H~\mbox{and}~y^j\in\mathcal{O}(\pi),\\
	(y^{0:H},j)\models_{\rm{S}}\lnot\phi~\mbox{iff}~ & (y^{0:H},j)\not\models_{\rm{W}}\phi,\\
	(y^{0:H},j)\models_{\rm{S}}\phi_{1}\wedge\phi_{2}~\mbox{iff}~ &  (y^{0:H},j)\models_{\rm{S}}\phi
	_{1}~\\& ~\mbox{and}~(y^{0:H},j)\models_{\rm{S}}\phi_{2},\\
	(y^{0:H},j)\models_{\rm{S}}\phi_{1}\mathcal{U}_{\mathcal{I}}\phi_{2}~\mbox{iff}~ &  \exists
	j^{\prime}\in j+\mathcal{I}, \mbox{s.t.} (y^{0:H},j^{\prime})\models_{\rm{S}}\phi_{2},\\
	&   (y^{0:H},j^{\prime\prime})\models_{\rm{S}}\phi_{1} \forall j^{\prime\prime}\in\lbrack j,j^{\prime}).
	\end{split}
	\]
	\label{strong}
\end{definition}

\begin{definition}
	The Boolean semantics of MTL for trajectories of finite length in the weak view is defined recursively as follows~\cite{zhe_advisory}:
	\[
	\begin{split}
	(y^{0:H},j)\models_{\rm{W}}\pi~\mbox{iff}~& \textrm{either of the following holds}:\\
	& 1)~j\le H~\mbox{and}~y^j\in\mathcal{O}(\pi);\\
	& 2)~j>H,\\	
	(y^{0:H},j)\models_{\rm{W}}\lnot\phi~\mbox{iff}~ & (y^{0:H},j)\not\models_{\rm{S}}\phi,\\	
	(y^{0:H},j)\models_{\rm{W}}\phi_{1}\wedge\phi_{2}~\mbox{iff}~ &  (y^{0:H},j)\models_{\rm{W}}\phi
	_{1}~\\&~\mbox{and}~(y^{0:H},j)\models_{\rm{W}}\phi_{2},\\	
	(y^{0:H},j)\models_{\rm{W}}\phi_{1}\mathcal{U}_{\mathcal{I}}\phi_{2}~\mbox{iff}~ & \exists
	j^{\prime}\in j+\mathcal{I}, \mbox{s.t.} (y^{0:H},j^{\prime})\models_{\rm{W}}\phi_{2},\\
	&  (y^{0:H},j^{\prime\prime})\models_{\rm{W}}\phi_{1} \forall j^{\prime\prime}\in\lbrack j, j^{\prime}).
	\label{weak}
	\end{split}
	\]
\end{definition}

Intuitively, if a trajectory of finite length can be extended to infinite length, then the strong view indicates that the truth value of the formula on the infinite-length trajectory is already ``determined'' on the trajectory of finite length, while the weak view indicates that it may not be ``determined'' yet \cite{Ho2014}. As an example, a trajectory $y^{0:3}=y^0y^1y^2y^3$ is not possible to strongly satisfy $\phi=\Box_{[0,5]}\pi$ at time 0, but $y^{0:3}$ is possible to strongly violate $\phi$ at time 0, i.e., $(y^{1:3},0)\models_{\rm{S}}\lnot\phi$ is possible.

For an MTL formula $\phi$, the necessary length $\norm{\phi}$ is defined recursively as follows \cite{Maler2004}: 
\[                                                                   
\begin{split}
&\norm{\pi} =0, ~\norm{\lnot\phi} =\norm{\phi},\\
&\norm{\phi_{1}\wedge\phi_{2}}=\max(\norm{\phi_{1}},\norm{\phi_{2}}),\\
&\norm{\phi_1\mathcal{U}_{[j_1, j_2]}\phi_{2}}=\max(\norm{\phi_{1}},\norm{\phi_{2}})+j_2.
\end{split}
\]   

\subsection{Problem Statement}
\label{sec_problem}
We now present the problem formulation for the control of the MAS with intermittent communication and MTL specifications.

\begin{problem}
	Design the control inputs for the leader $\mathbf{u}_0 = [u^0_0, u^1_0, \cdots]$ ($u^j_0$ denotes the control input at time index $j$) such that the following characteristics are satisfied while minimizing the control effort $\norm{\mathbf{u}_0}$\footnote{$\left\Vert\cdot\right\Vert $ denotes the 2-norm.}:\\
	\textit{Correctness}: A given MTL specification $\phi$ is weakly satisfied by the trajectory of the leader.\\
	\textit{Stability}: The error signal $e_{1,i}\left(t\right)$ is uniformly bounded, and the error signal $e_{2,i}\left(t\right)$ is asymptotically regulated\footnote{The error signal $e_{2,i}\left(t\right)$ is asymptotically regulated
		if $\left\Vert e_{2,i}\left(t\right)\right\Vert \rightarrow0$ as
		$t\rightarrow\infty.$} for each follower $i$.\\
	\textit{Consensus}: The states of the followers asymptotically reach consensus to $x_{\textrm{g}}$.
	\label{problem}
\end{problem}

\section{Stability and Consensus Analysis}
\label{sec_dwell_time}
In this section, we provide the conditions for achieving the stability of the switched system and the consensus of the followers. Such conditions include maximal (see Theorem \ref{Theorem 1}) and minimal (see Theorem \ref{Theorem 2}) dwell-time conditions on the intervals between consecutive time instants at which the leader should provide state information to the same follower.

\begin{theorem}
	\label{Theorem 1} 
	Let $V_{\textrm{T}}\in\mathbb{R}_{>0}$ be a user-defined
	parameter. Then, the error signal in (\ref{e1i}) for follower $i$
	is uniformly bounded, i.e., $\left\Vert e_{1,i}\left(t\right)\right\Vert \leq V_{\textrm{T}}$
	for all $t\geq0$, provided the leader satisfies the maximum dwell-time
	condition 
	\begin{alignat}{1}
	t_{s+1}^{i}-t_{s}^{i} & \leq\frac{1}{\lambda_{\textrm{max}}\left(A\right)}\ln\left(\frac{\lambda_{\textrm{max}}\left(A\right)V_{\textrm{T}}}{\overline{d}_{i}}+1\right)
	\label{Maximum Dwell-Time Condition}
	\end{alignat}
	for all $s\in\mathbb{Z}_{\geq0}$.
\end{theorem}

\begin{proof}
	Let $s\in\mathbb{Z}_{\geq0}.$ Consider the common Lyapunov functional
	candidate $V_{1,i}:\mathbb{R}^{m}\rightarrow\text{\ensuremath{\mathbb{R}}}_{\geq0}$
	\begin{equation}
	V_{1,i}\left(e_{1,i}\left(t\right)\right)\triangleq\frac{1}{2}e_{1,i}^{T}\left(t\right)e_{1,i}\left(t\right).
	\label{V1}
	\end{equation}
	By (\ref{Follower Dynamics}) and (\ref{Reset}),
	(\ref{e1i}) is continuously differentiable over $\left[t_{s}^{i},t_{s+1}^{i}\right).$
	Substituting (\ref{e1i Dot closed-loop}) when $t\in\left[t_{s}^{i},t_{s+1}^{i}\right)$
	into the time-derivative of (\ref{V1}) yields $\dot{V}_{1,i}\left(e_{1,i}\left(t\right)\right)=e_{1,i}^{T}\left(t\right)\left(Ae_{1,i}\left(t\right)-d_{i}\left(t\right)\right),$
	which can be upper bounded by 
	\begin{equation}
	\dot{V}_{1,i}\left(e_{1,i}\left(t\right)\right)\leq\lambda_{\textrm{max}}\left(A\right)\left\Vert e_{1,i}\left(t\right)\right\Vert ^{2}+\overline{d}_{i}\left\Vert e_{1,i}\left(t\right)\right\Vert \label{V1 Dot Bound}
	\end{equation}
	Substituting (\ref{V1}) into (\ref{V1 Dot Bound})
	produces
	\begin{equation}
	\dot{V}_{1,i}\left(e_{1,i}\left(t\right)\right)\leq2\lambda_{\textrm{max}}\left(A\right)V_{1,i}\left(e_{1,i}\left(t\right)\right)+\overline{d}_{i}\sqrt{2V_{1,i}\left(e_{1,i}\left(t\right)\right)}.\label{V1 Dot Bound 2}
	\end{equation}
 
	Invoking the Comparison Lemma \cite[Lemma 3.4]{Khalil} on 
	(\ref{V1 Dot Bound 2}) over $\left[t_{s}^{i},t_{s+1}^{i}\right)$
	yields 
	\begin{equation}
	V_{1,i}\left(e_{1,i}\left(t\right)\right)\leq\left(\frac{\overline{d}_{i}\sqrt{2}}{2\lambda_{\textrm{max}}\left(A\right)}\left(e^{\lambda_{\textrm{max}}\left(A\right)\left(t-t_{s}^{i}\right)}-1\right)\right)^{2}.\label{V1 Bound}
	\end{equation}
	Substituting (\ref{V1}) into (\ref{V1 Bound}) yields $\left\Vert e_{1,i}\left(t\right)\right\Vert \leq\frac{\overline{d}_{i}}{\lambda_{\textrm{max}}\left(A\right)}\left(e^{\lambda_{\textrm{max}}\left(A\right)\left(t-t_{s}^{i}\right)}-1\right).$
	Now, define $\Phi_{i}:\left[t_{s}^{i},t_{s+1}^{i}\right]\rightarrow\mathbb{R}_{\geq0}$
	by $\Phi_{i}\left(t\right)\triangleq\frac{\overline{d}_{i}}{\lambda_{\textrm{max}}\left(A\right)}\left(e^{\lambda_{\textrm{max}}\left(A\right)\left(t-t_{s}^{i}\right)}-1\right).$
	Since $\left\Vert e_{1,i}\left(t\right)\right\Vert \leq\frac{\overline{d}_{i}}{\lambda_{\textrm{max}}\left(A\right)}\left(e^{\lambda_{\textrm{max}}\left(A\right)\left(t-t_{s}^{i}\right)}-1\right)$
	for all $t\in\left[t_{s}^{i},t_{s+1}^{i}\right)$ and $\left\Vert e_{1,i}\left(t_{s+1}^{i}\right)\right\Vert =0$
	where $\Phi_{i}\left(t_{s+1}^{i}\right)>0,$ then $\left\Vert e_{1,i}\left(t\right)\right\Vert \leq\Phi_{i}\left(t\right)$
	for all $t\in\left[t_{s}^{i},t_{s+1}^{i}\right].$ If $\Phi_{i}\left(t_{s+1}^{i}\right)\leq V_{\textrm{T}},$
	then $\left\Vert e_{1,i}\left(t\right)\right\Vert \leq V_{\textrm{T}}$ for
	all $t\in\left[t_{s}^{i},t_{s+1}^{i}\right].$ Hence, the corresponding
	dwell-time condition is given by (\ref{Maximum Dwell-Time Condition}).
	Since $\left[0,\infty\right)=\underset{s\in\mathbb{Z}_{\geq0}}{\bigcup}\left[t_{s}^{i},t_{s+1}^{i}\right)$                                         
	where $t_{0}^{i}=0$ and $\left\Vert e_{1,i}\left(t\right)\right\Vert \leq V_{\textrm{T}}$
	over each $\left[t_{s}^{i},t_{s+1}^{i}\right)$ provided the leader
	continuously satisfies the dwell-time condition in (\ref{Maximum Dwell-Time Condition}),                                            
	then $\left\Vert e_{1,i}\left(t\right)\right\Vert \leq V_{\textrm{T}}$ for
	all $t\in\left[0,\infty\right)$.
\end{proof}

\begin{theorem}
        The error signal in (\ref{e2i})
		is globally asymptotically regulated provided the leader satisfies 
		both the maximum dwell-time condition in (\ref{Maximum Dwell-Time Condition})
		and the minimum dwell-time condition in
		\begin{equation}
		t_{s+1}^{i}-t_{s}^{i}>\frac{1}{k_{i}}\ln\left(\frac{\left\Vert e_{2,i}\left(t_{s}^{i}\right)\right\Vert }{\left\Vert e_{2,i}\left(t_{s}^{i}\right)\right\Vert -V_{\textrm{T}}}\right)
		\label{Minimum Dwell-Time Condition}
		\end{equation}
		for all $s\in\mathbb{Z}_{>0}$ such that $s<\bar{s}$ ($\bar{s}$
		denotes the index of $t_{\bar{s}}^{i}$ where $\left\Vert e_{2,i}\left(t_{\bar{s}}^{i}\right)\right\Vert \leq V_{\textrm{T}}$
		first holds), and $V_{\textrm{T}}\in\Big(0,\frac{R_{\textrm{g}}}{2\lambda_{\textrm{max}}(C)}\Big]$.
	\label{Theorem 2} 
\end{theorem}

\begin{proof}
	Suppose the leader satisfies the dwell-time condition in (\ref{Maximum Dwell-Time Condition})
	for all $s\in\mathbb{Z}_{\geq0}$. Consider the common Lyapunov functional
	$V_{2,i}:\mathbb{R}^{m}\rightarrow\text{\ensuremath{\mathbb{R}}}_{\geq0}$
	\begin{equation}
	V_{2,i}\left(e_{2,i}\left(t\right)\right)\triangleq\frac{1}{2}e_{2,i}^{T}\left(t\right)e_{2,i}\left(t\right).
	\label{V2}
	\end{equation}
	By (\ref{Follower Dynamics}) and (\ref{Reset}),
	(\ref{e2i}) is continuously differentiable over $\left[t_{s}^{i},t_{s+1}^{i}\right)$. 
	Substituting (\ref{e2i Dot closed-loop}) when $\text{ }t\in\left[t_{s}^{i},t_{s+1}^{i}\right)$
	into the time-derivative of (\ref{V2}) yields 
	\begin{equation}
	\dot{V}_{2,i}\left(e_{2,i}\left(t\right)\right)=-k_{i}e_{2,i}^{T}\left(t\right)e_{2,i}\left(t\right)\label{V2 Dot}
	\end{equation}
	where substituting (\ref{V2}) into (\ref{V2 Dot}) yields
	\begin{equation}
	\dot{V}_{2,i}\left(e_{2,i}\left(t\right)\right)=-2k_{i}V_{2,i}\left(e_{2,i}\left(t\right)\right).\label{V2 Dot 2}
	\end{equation}
	The solution of (\ref{V2 Dot 2}) over $\left[t_{s}^{i},t_{s+1}^{i}\right)$
	is given by $V_{2,i}\left(e_{2,i}\left(t\right)\right)=V_{2,i}\left(e_{2,i}\left(t_{s}^{i}\right)\right)e^{-2k_{i}\left(t-t_{s}^{i}\right)}$
	where substituting (\ref{V2}) results in 
	\begin{equation}
	\left\Vert e_{2,i}\left(t\right)\right\Vert =\left\Vert e_{2,i}\left(t_{s}^{i}\right)\right\Vert e^{-k_{i}\left(t-t_{s}^{i}\right)}.\label{e2i Solution}
	\end{equation}
	Observe that $e_{2,i}\left(t_{s}^{i}\right)$ is finite since $x_{i}\left(t_{s}^{i}\right)$
	is a measured quantity provided by the leader where (\ref{e2i Solution})
	implies $e_{2,i}\left(t\right)$ is bounded over $\left[t_{s}^{i},t_{s+1}^{i}\right).$
	Moreover, the RHS of follower $i's$ dynamics in (\ref{Follower Dynamics})
	are Lebesgue measurable and locally essentially bounded. Therefore,
	there exists a Filippov solution $x_{i}\left(t\right)$ that is absolutely
	continuous over $\left[0,\infty\right).$ Now, consider $t\in\left[t_{s}^{i},t_{s+1}^{i}\right).$
	The jump discontinuity of $e_{2,i}\left(t\right)$ at $t_{s+1}^{i}$
	is given by $\Omega_{i}\left(t_{s+1}^{i}\right)\triangleq e_{2,i}\left(t_{s+1}^{i}\right)-\underset{t\rightarrow\left(t_{s+1}^{i}\right)^{-}}{\mathrm{lim}}e_{2,i}\left(t\right)$
	where $e_{2,i}\left(t_{s+1}^{i}\right)$ is defined by (\ref{e2i Reset})
	and $\underset{t\rightarrow\left(t_{s+1}^{i}\right)^{-}}{\mathrm{lim}}e_{2,i}\left(t\right)$
	denotes the limit of $e_{2,i}\left(t\right)$ as $t\rightarrow t_{s+1}^{i}$
	from the left. Since $\Omega_{i}\left(t_{s+1}^{i}\right)=\underset{t\rightarrow\left(t_{s+1}^{i}\right)^{-}}{\mathrm{lim}}e_{1,i}\left(t\right)$
	and $\left\Vert \cdot\right\Vert $ is continuous over $\mathbb{R}$,
	then by Theorem \ref{Theorem 1} $\left\Vert \Omega_{i}\left(t_{s+1}^{i}\right)\right\Vert \leq V_{\textrm{T}}$. It then follows that the magnitude of the jump discontinuity is
	bounded by 
	\begin{equation}
	\left|\left\Vert e_{2,i}\left(t_{s+1}^{i}\right)\right\Vert -\underset{t\rightarrow\left(t_{s+1}^{i}\right)^{-}}{\mathrm{lim}}\left\Vert e_{2,i}\left(t\right)\right\Vert \right|\leq V_{\textrm{T}}.\label{Jump Discontinuity Magntiude}
	\end{equation}
Since $\left\Vert e_{2,i}\left(t\right)\right\Vert $ is strictly
decreasing over $\left[t_{s}^{i},t_{s+1}^{i}\right)$ by (\ref{e2i Solution}),
then $\left\Vert e_{2,i}\left(t\right)\right\Vert \leq\left\Vert e_{2,i}\left(t_{s}^{i}\right)\right\Vert $
for all $t\in\left[t_{s}^{i},t_{s+1}^{i}\right).$ The reset map in
(\ref{Reset}) may induce an instantaneous growth in (\ref{e2i})
at $t_{s+1}^{i}$ where (\ref{Jump Discontinuity Magntiude}) implies
$\left\Vert e_{2,i}\left(t_{s+1}^{i}\right)\right\Vert \leq V_{\textrm{T}}+\left\Vert e_{2,i}\left(t_{s}^{i}\right)\right\Vert e^{-k_{i}\left(t_{s+1}^{i}-t_{s}^{i}\right)}.$
Therefore, the minimum dwell-time condition given by  (\ref{Minimum Dwell-Time Condition})
can ensure that $\left\Vert e_{2,i}\left(t_{s}^{i}\right)\right\Vert >\left\Vert e_{2,i}\left(t_{s+1}^{i}\right)\right\Vert$,
which is valid when $\left\Vert e_{2,i}\left(t_{s}^{i}\right)\right\Vert >V_{\textrm{T}}>0$.
\textcolor{black}{Observe that there exists some $t_{\bar{s}}^{i}\in\mathbb{R}_{>0}$ such that $\left\Vert e_{2,i}\left(t_{\bar{s}}^{i}\right)\right\Vert \leq V_{\textrm{T}}$.
    Provided the leader
	satisfies the maximum dwell-time condition in (\ref{Maximum Dwell-Time Condition})
	for all $t\le t_{\bar{s}}^{i}$, then $\left\Vert Cx_{\textrm{g}}-y_{i}\left(t_{\bar{s}}^{i}\right)\right\Vert \leq\lambda_{\textrm{max}}(C)\left\Vert e_{2,i}\left(t_{\bar{s}}^{i}\right)\right\Vert +\lambda_{\textrm{max}}(C)\left\Vert e_{1,i}\left(t_{\bar{s}}^{i}\right)\right\Vert \leq2\lambda_{\textrm{max}}(C)V_{\textrm{T}}.$
	Hence, by selecting $V_{\textrm{T}}\in\Big(0,\frac{R_{\textrm{g}}}{2\lambda_{\textrm{max}}(C)}\Big]$, it
	follows that $\left\Vert Cx_{\textrm{g}}-y_{i}\left(t_{\bar{s}}^{i}\right)\right\Vert \leq R_{\textrm{g}}$, and 
	 follower $i$ will be inside the feedback region after $t_{\bar{s}}^{i}$. Moreover, $\left\Vert e_{1,i}\left(t\right)\right\Vert =0$
	and $\left\Vert e_{2,i}\left(t\right)\right\Vert =\left\Vert e_{2,i}\left(t_{\bar{s}}^{i}\right)\right\Vert e^{-k_{i}\left(t-t_{\bar{s}}^{i}\right)}$
	for all $t\geq t_{\bar{s}}^{i}$. Thus, $\left\Vert e_{2,i}\left(t\right)\right\Vert \rightarrow0$
	as $t\rightarrow\infty.$} Since (\ref{V2}) does not have a restricted
domain and is radially unbounded, then the stability result is global.
\end{proof} 

\begin{remark}
	The proof of Theorem \ref{Theorem 2} formally excludes Zeno behavior.          
\end{remark}

\begin{remark}
	From Theorem \ref{Theorem 1} and Theorem \ref{Theorem 2}, for stability and consensus, for any $i$ and $s$,
	\begin{align}
	\begin{split}
	&\frac{1}{\lambda_{\textrm{max}}\left(A\right)}\ln\left(\frac{\lambda_{\textrm{max}}\left(A\right)V_{\textrm{T}}}{\overline{d}_{i}}+1\right)\ge
	\frac{1}{k_{i}}\ln\left(\frac{\left\Vert e_{2,i}\left(t_{s}^{i}\right)\right\Vert }{\left\Vert e_{2,i}\left(t_{s}^{i}\right)\right\Vert -V_{\textrm{T}}}\right).
	\end{split}
	\end{align}
\end{remark}

With Theorem \ref{Theorem 1} and Theorem \ref{Theorem 2}, we provide the following theorem for achieving consensus of the followers.

\begin{theorem}
	\label{Theorem 3} The states of the followers asymptotically reach consensus to $x_{\textrm{g}}$ if the maximum dwell-time condition in (\ref{Maximum Dwell-Time Condition}) and the minimum dwell-time condition in (\ref{Minimum Dwell-Time Condition}) are satisfied or all $t_{s}^{i}\le t_{\bar{s}}^{i}$ ($i\in F$), and $V_{\textrm{T}}\in\Big(0,\frac{R_{\textrm{g}}}{2\lambda_{\textrm{max}}(C)}\Big]$.
\end{theorem}	

\begin{proof} 
Let $i\in F$. By Theorem \ref{Theorem 1}, if the maximum dwell-time condition in (\ref{Maximum Dwell-Time Condition}) is satisfied, then
$\left\Vert e_{1,i}\left(t\right)\right\Vert \leq V_{\textrm{T}}$ for all                           
$t\geq0$. By Theorem \ref{Theorem 2}, if the minimum dwell-time condition in (\ref{Minimum Dwell-Time Condition}) is satisfied or all $t_{s}^{i}\le t_{\bar{s}}^{i}$ ($i\in F$), then there exists a time $T_{i}\in\mathbb{R}_{>0}$ such that $\left\Vert e_{2,i}\left(T_{i}\right)\right\Vert \leq V_{\textrm{T}}$. Therefore, $\left\Vert Cx_{\textrm{g}}-y_{i}\left(T_{i}\right)\right\Vert \leq\lambda_{\textrm{max}}(C)\left\Vert e_{1,i}\left(T_{i}\right)\right\Vert +\lambda_{\textrm{max}}(C)\left\Vert e_{2,i}\left(T_{i}\right)\right\Vert \leq 2\lambda_{\textrm{max}}(C)V_{\textrm{T}}\le R_{\textrm{g}}$ as $V_{\textrm{T}}\in\Big(0,\frac{R_{\textrm{g}}}{2\lambda_{\textrm{max}}(C)}\Big]$.
Then for $t\ge T_{i}$, follower $i$ will be inside the feedback region where $\left\Vert e_{1,i}\left(t\right)\right\Vert =0$. Moreover,
$\left\Vert e_{2,i}\left(t\right)\right\Vert\rightarrow0$
as $t\rightarrow\infty$, so $\left\Vert x_{\textrm{g}}-x_{i}\left(t\right)\right\Vert=\left\Vert e_{1,i}\left(t\right)\right\Vert+\left\Vert e_{2,i}\left(t\right)\right\Vert=\left\Vert e_{2,i}\left(t\right)\right\Vert\rightarrow0$
as $t\rightarrow\infty$.
\end{proof}

\section{Controller Synthesis with Intermittent Communication and MTL Specifications}
\label{sec_reactive}
In this section, we provide the framework and algorithms for controller synthesis of the leader to satisfy the maximum and minimal dwell-time conditions and the practical constraints. The controller synthesis is conducted iteratively as the state estimates for the followers are reset to the true state values whenever they are serviced by the leader, and thus the control inputs need to be recomputed with the reset values.

We assume that the communication is only possible at discrete time instants, with $T_{\textrm{s}}$ time periods apart and controlled by the communication switching signal $\zeta_i$. We define the discrete time set $\mathbb{T}_{\textrm{d}}\triangleq\{t[0], t[1], \dots\}$, where $t[j]=jT_{\textrm{s}}$ for $j\in\mathbb{I}$. The maximum dwell-time $\frac{1}{\lambda_{\textrm{max}}\left(A\right)}\ln\left(\frac{\lambda_{\textrm{max}}\left(A\right)V_{\textrm{T}}}{\overline{d}_{i}}+1\right)$ in (\ref{Maximum Dwell-Time Condition}) for robot $i$ $(i=1,\dots,Q)$ is in the interval $[n_iT_s, (n_i+1)T_{\textrm{s}})$ and the minimum dwell-time $\frac{1}{k_{i}}\ln\left(\frac{\left\Vert e_{2,i}\left(t_{0}^{i}\right)\right\Vert }{\left\Vert e_{2,i}\left(t_{0}^{i}\right)\right\Vert -V_{\textrm{T}}}\right)$ in (\ref{Minimum Dwell-Time Condition}) is in the interval $[(m_i-1)T_{\textrm{s}}, m_iT_s)$. 
We use the following MTL specifications for encoding the maximum dwell-time condition and the minimum dwell-time condition ($\eta\in[0, R)$ is a user-defined parameter):
\begin{align}
\begin{split}
&\phi_1=\bigwedge_{1\le i\le Q}\big(\Box\Diamond_{[0, n_i]}\norm{y_0-\hat{y}_i}\le\eta\big),\\
&\phi_2=\bigwedge_{1\le i\le Q}\big(\Box(\norm{y_0-\hat{y}_i}\le\eta\Rightarrow \Box_{[1, m_i]}\norm{y_0-\hat{y}_i}>\eta)\big),
\end{split}
\end{align} 
where $\phi_1$ means ``for any follower $i$, the leader needs to be within $\eta$ distance from the estimated position of the follower $i$ at least once in any $n_iT_s$ time periods'', and $\phi_2$ means ``each time the leader is within $\eta$ distance from the estimated position of the follower $i$, it should not be within $\eta$ distance from the estimated position of the follower $i$ again for the next $m_iT_s$ time periods''. 

The leader also needs to satisfy an MTL specification $\phi_{\textrm{p}}$ for the practical constraints. One example of $\phi_{\textrm{p}}$ is as follows:
\begin{align}
&\phi_{\textrm{p}}=\Box\Diamond_{[0, c]}\big((y_0\in G_1)\vee (y_0\in G_2)\big)\wedge\Box (y_0\in D).
\end{align} 
which means ``the leader robot needs to reach the charging station $G_1$ or $G_2$ at least once in any $cT_{\textrm{s}}$ time periods, and it should always remain in the region $D$''.

Combining $\phi_1$, $\phi_2$ and $\phi_{\textrm{p}}$, the MTL specification for the leader is                                  
$\phi=\phi_1\wedge\phi_2\wedge\phi_{\textrm{p}}$.

We use $[\phi]^{\ell}_{j}$ to denote the formula modified from the MTL formula $\phi$ when $\phi$ is evaluated at time index $j$ and the current time index is $\ell$.  $[\phi]^{\ell}_{j}$ can be calculated recursively as follows (we use $\pi_{j}$ to denote the atomic predicate $\pi$ evaluated at time index $j$):
\begin{align}
\begin{split}                     
[\pi]^{\ell}_{j} =&  
\begin{cases}
\pi_{j},& \mbox{if $j>\ell$}\\  	
\top,& \mbox{if $j\le \ell$ and $y^{j}\in\mathcal{O}(\pi)$}\\  
\bot,& \mbox{if $j\le \ell$ and $y^{j}\not\in\mathcal{O}(\pi)$}                                                                                                                                         
\end{cases}\\
[\neg\phi]^{\ell}_{j}  :=&\neg[\phi]^{\ell}_{j}\\
[\phi_1\wedge\phi_2]^{\ell}_{j}:=&[\phi_1]^{\ell}_{j}\wedge[\phi_2]^{\ell}_{j}\\
[\phi_1\mathcal{U}_{\mathcal{I}}\phi_2]^{\ell}_{j} :=&\bigvee_{j'\in (j+\mathcal{I})}\Big([\phi_2]^{\ell}_{j'}\wedge\bigwedge_{j\le j''<j'}[\phi_1]^{\ell}_{j''}\Big).
\end{split}
\label{update_phi}
\end{align}
If the MTL formula $\phi$ is evaluated at the initial time index (which is the usual case when the task starts at the initial time), then the modified formula is $[\phi]^{\ell}_{0}$. 

Algorithm 1 shows the controller synthesis approach with intermittent communication and MTL specifications. 
The controller synthesis problem can be formulated as a sequence of mixed integer linear programming (MILP) problems, denoted as MILP-sol in Line \ref{MILP-sol} and expressed as follows:
\begin{align}
\underset{\mathbf{u}^{\ell}_0}{\argmin} ~ & J(\mathbf{u}^{\ell}_0)=\norm{\mathbf{u}^{\ell}_0}   \\
\text{subject to:} ~ 
& x_0^{j+1}=\bar{A}_0x_0^{j}+\bar{B}^0u_0^j, ~y_0^{j}=\bar{C}_0x_0^{j}, \nonumber \\ & ~~~\forall i=1,\dots,Q, \forall j=\ell,\dots,\ell+N-1,\\
& \hat{x}_i^{j+1}=\bar{A}x_i^{j}+\bar{B}u_i^j, ~\hat{y}_i^{j}=\bar{C}\hat{x}_i^{j}, \nonumber \\ & ~~~ \forall i=1,\dots,Q, \forall j=\ell,\dots,\ell+N-1, \label{update_constraint}\\
& u_{0,\textrm{min}}\le u_0^j\le u_{0,\textrm{max}}, \forall i=1,\dots,Q, \nonumber \\ & ~~~~~~~~~~~~~~\forall j=\ell,\dots,\ell+N,\\
& (\tilde{y}^{\ell:\ell+N-1}, 0)\models_{\textrm{W}} [\phi]^{\ell}_{0}, 
\label{MILP}               
\end{align} 
where the time index $\ell$ is initially set as 0, $N\in\mathbb{Z}_{>0}$ is the number of time instants in the control horizon, $\tilde{y}^{\ell:\ell+N-1}=[y_0^{\ell:\ell+N-1}, \hat{y}_1^{\ell:\ell+N-1}, \dots, \hat{y}_Q^{\ell:\ell+N-1}]$, \(\mathbf{u}^{\ell}_0 = [u^{\ell}_0, u^{\ell+1}_0, \cdots, u^{\ell+N-1}_0]\) is the control input signal of the leader, the input values are constrained to $[u_{0, \textrm{min}}, u_{0, \textrm{max}}]$, $\bar{A}^0$, $\bar{B}^0$, $\bar{C}^0$, $\bar{A}$, $\bar{B}$ and $\bar{C}$ are converted from $A^0$, $B^0$, $C^0$, $A$, $B$ and $C$ respectively for the discrete-time state-space representation, and $u_i^j$ are follower control inputs from (\ref{Follower Control}). Note that we only require the trajectory $y_0^{\ell:\ell+N-1}$ to weakly satisfy $\phi$ as $\ell+N-1$ may be less than the necessary length $\norm{\phi}$.

At each time index $\ell$, we check if there exists any follower that is being serviced (Line \ref{check_serve}). If there are such followers, we update the state estimates of those followers with their true state values (Line \ref{update}). Then we modify the MTL formula as in (\ref{update_phi}) and the updated $m_i$ (Line \ref{update_MTL}). The MILP is solved for time $\ell$ with the updated state values and the modified MTL formula $[\phi]^{\ell}_{0}$ (Line \ref{recompute}). The previously computed leader control inputs are replaced by the newly computed control inputs from time index $\ell$ to $\ell+N-1$ (Line \ref{replace}).

\begin{algorithm}[tp]
	\caption{Controller synthesis of MASs with intermittent communication and MTL specifications.}                                                                   
	\label{MTLalg}
	\begin{algorithmic}[1]
		\State \textbf{Inputs:}  $x_0^{0}$, $x_i^{0}$, $\phi$, $x_{\textrm{g}}$, $R_{\textrm{g}}$, $V_{\textrm{T}}$, $\eta$, $T_s$, $k_i$
		\State $\ell\gets0$
		\State Solve MILP-sol to obtain the optimal inputs $u_{0}^{\ast q}~(q=0,1,\dots,N-1)$ \label{MILP-sol}
		\While{$\norm{Cx_{\textrm{g}}-y_i(t[\ell])}>R_{\textrm{g}}$ for some $i$}
		\State $\mathcal{W}=\{i~\vert~\norm{y_0-\hat{y}_i(t[\ell])}\le\eta\}$ \label{check_serve}
		\If{$\vert\mathcal{W}\vert\neq\emptyset$}
		\State 	$\forall i\in\mathcal{W}$, update $\hat{x}_i^{\ell}$ in constraint (\ref{update_constraint}) and change constraint (\ref{update_constraint}) as follows: \label{update}
		\[
		\begin{split}
		& \hat{x}_i^{j+1}=\bar{A}x_i^{j}+\bar{B}u_i^j, \forall i=1,\dots,Q, \nonumber \\ & ~~~~~~~~~~~~~~~~~~~~~~~~\forall j=\ell,\ell+1,\dots,\ell+N-1,\\
		& \hat{x}_i^{\ell}=x_i^{\ell}, \forall i\in\mathcal{W}
		\end{split}
		\]	
		\State Update $m_i$ in $\phi$ such that $\frac{1}{k_{i}}\ln\left(\frac{\left\Vert e_{2,i}\left(t[\ell]\right)\right\Vert }{\left\Vert e_{2,i}\left(t[\ell]\right)\right\Vert -V_{\textrm{T}}}\right)$ is in the interval $[(m_i-1)T_{\textrm{s}}, m_iT_s)$ \label{update_MTL}
		\State Re-solve MILP-sol to obtain the optimal inputs $u^{\ast\ell+q}~(q=0,1,\dots, N-1)$	\label{recompute}
		\State $u_{0}^{\ast\ell+q}\gets u^{\ast\ell+q}~(q=0,1,\dots, N-1)$				\label{replace}
		\EndIf		                                 
		\EndWhile    		                
		\State Return \(\mathbf{u}^{\ast0}=(u^{\ast0}_0, u^{\ast1}_0, \dots)\)
	\end{algorithmic}
\end{algorithm}

We use $\hat{t}_{s+1}^{i}$ to denote the $(s+1)^{th}$ time that $\norm{y_0(t)-y_i(t)}\le\eta$ holds in the discrete time set $\mathbb{T}_{\textrm{d}}$ for follower $i$\footnote{For $s=0,$ $\hat{t}_{0}^{i}$ is the initial time, i.e., $\hat{t}_{0}^{i}=0.$}, i.e.,
\begin{align}\nonumber
\hat{t}_{s+1}^{i}\triangleq &\inf\left\{ t\geq \hat{t}_{s}^{i}: (t\in\mathbb{T}_{\textrm{d}})\wedge\big(\left\Vert \hat{y}_{i}\left(t\right)-y_{0}\left(t\right)\right\Vert \leq \eta\big)\right\}.
\end{align}
We design the communication switching signal $\zeta_i$ as follows:
\begin{align}
\zeta_i(t)=&
\begin{cases}
1, ~~~~~~\mbox{if}~t=\hat{t}_{s}^{i}~\mbox{for~some}~s;   \\
0, ~~~~~~\mbox{otherwise}.
\end{cases}
\label{communication} 
\end{align}

Finally, we present Theorem \ref{Theorem 4} which provides theoretical guarantees for achieving correctness, stability and consensus (in Problem \ref{problem}).
\begin{theorem}
    With the observers in (\ref{Reset}),
	 follower controllers in (\ref{Follower Control}), communication switching signal in (\ref{communication}), if each optimization is feasible in Algorithm \ref{MTLalg} and $V_{\textrm{T}}\in\Big(0,  \min\{\frac{R_{\textrm{g}}}{2\lambda_{\textrm{max}}(C)}, \frac{R-\eta}{\lambda_{\textrm{max}}(C)}\}\Big]$ where $\eta\in[0, R)$,
	then Algorithm \ref{MTLalg} terminates within finite time, with the MTL specification $\phi$ weakly satisfied and the 
	followers asymptotically reaching consensus to the state $x_{\textrm{g}}$.
	\label{Theorem 4} 
\end{theorem}

\begin{proof}
We first use induction to prove that $\hat{t}_{s}^{i}=t_{s}^{i}$ holds for each $i$ and $s$. For each $i$, if $s=0$, then $\hat{t}_{0}^{i}=t_{0}^{i}=0$. Now assume that $\hat{t}_{s}^{i}=t_{s}^{i}$ holds and we prove that $\hat{t}_{s+1}^{i}=t_{s+1}^{i}$ holds. If each optimization is feasible in Algorithm \ref{MTLalg}, then $\hat{t}_{s+1}^{i}-\hat{t}_{s}^{i}=\hat{t}_{s+1}^{i}-t_{s}^{i}\le n_iT_{\textrm{s}} \leq\frac{1}{\lambda_{\textrm{max}}\left(A\right)}\ln\left(\frac{\lambda_{\textrm{max}}\left(A\right)V_{\textrm{T}}}{\overline{d}_{i}}+1\right)$. Then, following the analysis in the proof of Theorem \ref{Theorem 1}, we can derive that $\left\Vert e_{1,i}\left(\hat{t}_{s+1}^{i}\right)\right\Vert \leq V_{\textrm{T}}$. 
Thus, we have $\left\Vert y_{i}\left(\hat{t}_{s+1}^{i}\right)-y_{0}\left(\hat{t}_{s+1}^{i}\right)\right\Vert \leq \left\Vert Cx_{i}\left(\hat{t}_{s+1}^{i}\right)-C\hat{x}_{i}\left(\hat{t}_{s+1}^{i}\right)\right\Vert + \left\Vert C\hat{x}_{i}\left(\hat{t}_{s+1}^{i}\right)-Cx_{0}\left(\hat{t}_{s+1}^{i}\right)\right\Vert\le \lambda_{\textrm{max}}(C)V_{\textrm{T}}+\eta$. Therefore, if $V_{\textrm{T}}\le\frac{R-\eta}{\lambda_{\textrm{max}}(C)}$, then $\left\Vert y_{i}\left(\hat{t}_{s+1}^{i}\right)-y_{0}\left(\hat{t}_{s+1}^{i}\right)\right\Vert \leq R$. According to the communication switching signals in (\ref{communication}), we have $\zeta_i(\hat{t}_{s+1}^{i})=1$. Thus, from the definition of $t_{s+1}^{i}$ in Section \ref{sec_sensing}, we have $\hat{t}_{s+1}^{i}=t_{s+1}^{i}$ holds. Therefore, we have proven through induction that $\hat{t}_{s}^{i}=t_{s}^{i}$ hold for each $i$ and $s$. 
	
If each optimization is feasible in Algorithm \ref{MTLalg}, then the MTL specification $\phi$ is weakly satisfied. With $\hat{t}_{s}^{i}=t_{s}^{i}$, the maximum dwell-time condition in (\ref{Maximum Dwell-Time Condition}) and the minimum dwell-time condition in (\ref{Minimum Dwell-Time Condition}) are satisfied or all $t_{s}^{i}\le t_{\bar{s}}^{i}$ ($i\in F$). From Theorem \ref{Theorem 3}, if $V_{\textrm{T}}\in\Big(0,  \frac{R_{\textrm{g}}}{2\lambda_{\textrm{max}}(C)}\Big]$, then for each $i\in F$, there exists a time $T_i$ such that follower $i$ will be inside the feedback region for $t\ge T_{i}$. Thus, at time $\tilde{t}=\max_{i\in F}{T_i}$, $\norm{Cx_{\textrm{g}}-y_i(\tilde{t})}\le R_{\textrm{g}}$ holds for any $i\in F$, i.e., Algorithm \ref{MTLalg} is guaranteed to terminate within finite time.                                           
Finally, if $V_{\textrm{T}}\in\Big(0,  \min\{\frac{R_{\textrm{g}}}{2\lambda_{\textrm{max}}(C)}, \frac{R-\eta}{\lambda_{\textrm{max}}(C)}\}\Big]$, then the followers asymptotically reach consensus to $x_{\textrm{g}}$.
\end{proof}

\section{Implementation}
\label{sec_implementation}
We now demonstrate the controller synthesis approach on the example in Fig. \ref{fig_intro} (in Section \ref{sec_intro}). The leader is a quadrotor modeled as a three dimensional six degrees of freedom (6-DOF) rigid body \cite{zhe_advisory}. We denote the system state as $x^0_{\rm{q}}=[p_{\rm{q}},\dot{p}_{\rm{q}},\theta_{\rm{q}},\Omega_{\rm{q}}]^T\in\mathbb{R}^{12}$, where $p_{\rm{q}}=[x_{\rm{q},1},x_{\rm{q},2},x_{\rm{q},3}]^T$ and $\dot{p}_{\rm{q}}=[\dot{x}_{\rm{q},1},\dot{x}_{\rm{q},2},\dot{x}_{\rm{q},3}]^T$ are the position and velocity vectors of the quadrotor. The vector $\theta_{\rm{q}}=[\alpha_{\rm{q}},\beta_{\rm{q}},\gamma_{\rm{q}}]^T\in\mathbb{R}^{3}$ includes the roll, pitch and yaw Euler angles of the quadrotor. The vector $\Omega_{\rm{q}}\in\mathbb{R}^{3}$ includes the angular velocities rotating around its body frame axes. The general nonlinear dynamic model of such a quadrotor is given by                              
\begin{equation}\label{eqn_6DOFdynamics}       
\begin{array}{lll}                                
m_{\rm{q}}\ddot{p}_{\rm{q}}&=& r(\theta_{\rm{q}})T_{\rm{q}}\mathbf{e}_3-mg\mathbf{e}_3,\\
\dot{\theta}_{\rm{q}}&=&H(\theta_{\rm{q}})\Omega_{\rm{q}},\\
I\dot{\Omega}_{\rm{q}}&=&-\Omega_{\rm{q}}\times I\Omega_{\rm{q}}+\tau_{\rm{q}},
\end{array} 
\end{equation}
\textcolor{black}{where $m_{\rm{q}}$ is the mass, $g$ is the gravitational acceleration, $I$ is the inertia matrix, $r(\theta_{\rm{q}})$ is the rotation matrix representing the body frame with respect to the inertia frame (which is a function of the Euler angles), $H(\theta_{\rm{q}})$ is the nonlinear mapping matrix that projects the angular velocity $\Omega_{\rm{q}}$ to the Euler angle rate $\dot{\theta}_{\rm{q}}$, $\mathbf{e}_3=[0,0,1]^T$, $T_{\rm{q}}$ is the thrust of the quadrotor, and $\tau_{\rm{q}}\in\mathbb{R}^3$ is the torque on the three axes. The control input is $u_0=[u_{0, 1},u_{0, 2},u_{0, 3},u_{0, 4}]^T$, where $u_{0, 1}$ is the vertical velocity command, $u_{0, 2},u_{0, 3}$ and $u_{0, 4}$ are the angular velocity commands around its three body axes. The input values $u_{0, 1}, u_{0, 2}, u_{0, 3}$ and $u_{0, 4}$ are all bounded by $[-100, 100]$. By adopting the small-angle assumption and then linearizing the dynamic model around the hover state, a linear kinematic model can be obtained as follows:}
\begin{equation}
\begin{array}{ll}
\dot{x}_0=A_0x+B_0u_0, 
\end{array}
\end{equation}
where $x_0=[x_{\rm{q},1}, x_{\rm{q},2}, x_{\rm{q},3}, \dot{x}_{\rm{q},1}, \dot{x}_{\rm{q},2}, \alpha_{\rm{q}},\beta_{\rm{q}},\gamma_{\rm{q}}]$ is the state of the kinematic model of the quadrotor (leader), $A_0\in\mathbb{R}^{8\times8}$, and $B_0\in\mathbb{R}^{8\times4}$. For the 3-D position representation, $y_0=[x_{\rm{q},1}, x_{\rm{q},2}, x_{\rm{q},3}]^T$.

We use the following simplified dynamics for the followers
\begin{align}
\begin{aligned}
\dot{x}_{i,1} &= x_{i,1}+u_{i,1}+d_{i,1}, \\
\dot{x}_{i,2} &= x_{i,2}+u_{i,2}+d_{i,2},\\
\dot{x}_{i,3} &= 0,
\label{eq:sys-dyn-unicycle}
\end{aligned}
\end{align}                                                                                           
where $x_{i,1}$, $x_{i,2}$ and $x_{i,3}$ are the 3-D positions of follower $i$. Note that the vertical positions of the followers are constant.

For the state space representation, $x_i=[x_{i,1}, x_{i,2}, x_{i,3}]^T$, $u_i=[u_{i,1}, u_{i,2}, 0]^T$, $d_i=[d_{i,1}, d_{i,2}, 0]^T$ and $y_i=[x_{i,1}, x_{i,2}, x_{i,3}]^T$. The initial 3-D positions of the three followers are $[-20,-20,0]^T$, $[20,30,0]^T$ and $[40,-40,0]^T$, respectively. The initial 3-D position of the leader is $[-5,-30,5]^T$. The consensus state $x_{\textrm{g}}$ is set as $[0, 0,0]^T$. The random disturbances $d_{i}$ are bounded, i.e., $\left\Vert d_{i}\left(t\right)\right\Vert \leq\overline{d}_{i}$, where $\bar{d}_1=0.04, \bar{d}_2=0.03$ and $\bar{d}_3=0.02$.

For consensus, we consider the following control law from (\ref{Follower Control}):
\[
u_{i}\left(t\right)\triangleq -\hat{x}_i+k_{i}e_{2,i}\left(t\right),
\]
where $\hat{x}_{i}$ 
is the estimate of $x_{i}$, $k_1=0.1$, $k_2=0.15$ and $k_3=0.2$, respectively.

We consider two different scenarios with two different MTL specifications for the practical constraints.

\noindent\textit{Scenario 1}:\\
The leader needs to reach the charging station $G_1$ or $G_2$ at least once in any $6T_{\textrm{s}}$ time, and it should always remain in region $D$, where the two charging stations $G_1$ and $G_2$ are rectangular cuboids with length, width and height being 2, 2 and 5, centered at $[-20,10, 2.5]^T$ and $[25,0, 2.5]^T$, respectively, the region $D$ is a rectangular cuboid centered at $[0,0,7]^T$ with length, width and height being 30, 30 and 6, respectively (see Fig. \ref{fig_intro}). This specification is expressed as
\begin{align}\nonumber
&\phi^1_{\textrm{p}}=\Box\Diamond_{[0, 6]}\big((y_0\in G_1)\vee (y_0\in G_2)\big)\wedge\Box (y_0\in D).
\end{align} 

\noindent\textit{Scenario 2}:\\
The leader needs to reach the charging station $G_1$ or $G_2$ at least once in any $6T_{\textrm{s}}$ time, always remain in region $D$, and never stay in region $E$ for more than $2T_{\textrm{s}}$ time, where the region $E$ is a rectangular cuboid centered at $[0,0,6]^T$ with length, width and height being 15, 15 and 4, respectively. This specification is expressed as
\begin{align}\nonumber
\phi^2_{\textrm{p}}=&\Box\Diamond_{[0, 6]}\big((y_0\in G_1)\vee (y_0\in G_2)\big)\wedge\Box (y_0\in D)\\
&\wedge\lnot\Diamond\Box_{[0,2]} (y_0\in E).
\end{align} 

We set $R_{\textrm{g}}=R=5$, $V_{\textrm{T}}=1$,  $\eta=4$, $T_{\textrm{s}}=0.5$ and $N=20$. Fig. \ref{plot1} shows the simulation results in Scenario 1. The obtained input signals as shown in Fig. \ref{plot1} (a) gradually decrease as the followers approach $R_{\textrm{g}}$. Fig. \ref{plot1} (b) shows the 2-D planar plot of the trajectories of three followers
and a leader. $\Vert e_{i,1}(t)\Vert$ as shown in Fig. \ref{plot1} (c) is uniformly bounded by $V_{\textrm{T}}=1$. $\Vert e_{i,2}(t)\Vert$ as shown in Fig. \ref{plot1} (d) is monotonically decreasing when the followers approach consensus to $x_{\textrm{g}}$.  

Fig. \ref{plot2} shows the simulation results in Scenario 2. It can also be seen that $\Vert e_{i,1}(t)\Vert$ is uniformly bounded by $V_{\textrm{T}}=1$ and $\Vert e_{i,2}(t)\Vert$ is monotonically decreasing when the followers approach consensus to $x_{\textrm{g}}$. Note that with $\phi^2_{\textrm{p}}$, more control effort is needed to
satisfy the MTL specification after the followers arrive in region $E$ as the leader needs to get away from $E$ after each service to the followers.

\begin{figure}
	\centering
	\includegraphics[scale=0.2]{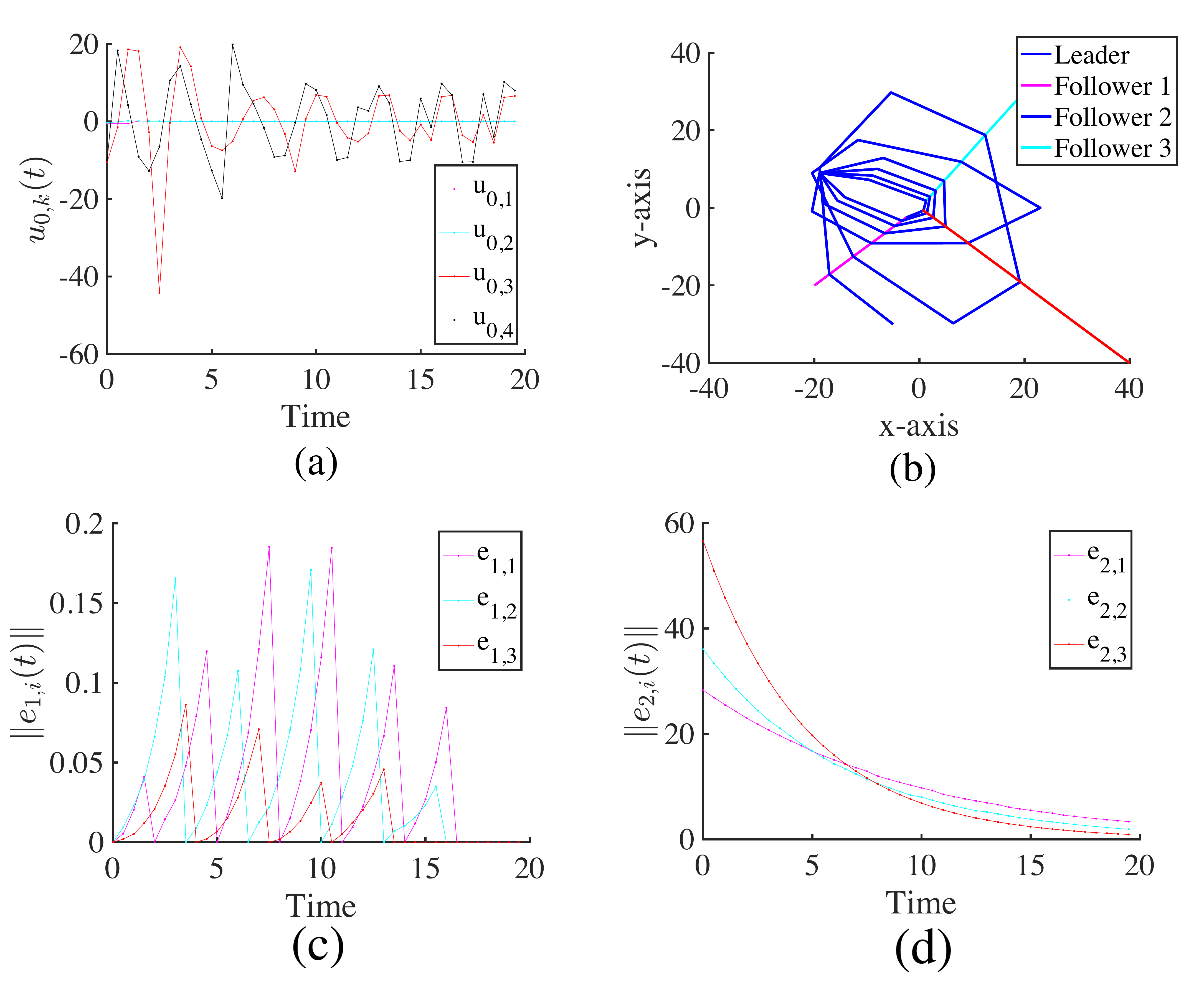}
	\caption{Results with MTL specification $\phi^1_{\textrm{p}}$ for the practical constraints: (a) the obtained optimal input signals; (b)
		2-D planar plot of the trajectories of three followers
		and a leader; (c) $\Vert e_{i,1}(t)\Vert$; (d) $\Vert e_{i,2}(t)\Vert$.}  
	\label{plot1}
\end{figure}

\begin{figure}
	\centering
	\includegraphics[scale=0.2]{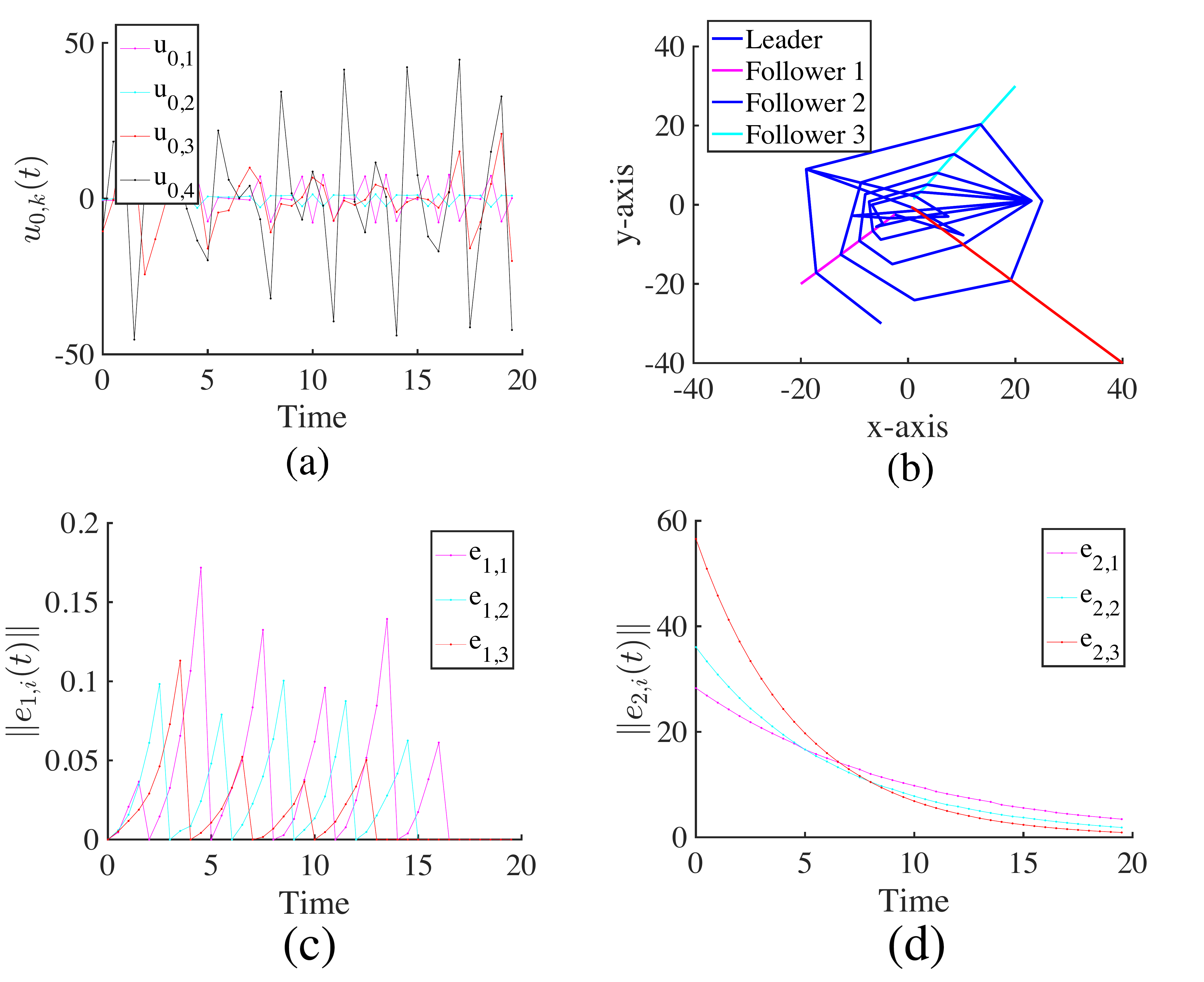}
	\caption{Results with MTL specification $\phi^2_{\textrm{p}}$ for the practical constraints: (a) the obtained optimal input signals; (b)
	2-D planar plot of the trajectories of three followers
	and a leader; (c) $\Vert e_{i,1}(t)\Vert$; (d) $\Vert e_{i,2}(t)\Vert$.}  
\label{plot2}
\end{figure}

\section{Conclusion}
\label{conclusion}
We presented a metric temporal logic approach for the controller synthesis of
a multi-agent system (MAS) with
intermittent communication. We iteratively solved a sequence of mixed-interger linear programmiung problems for provably achieving the correctness, stability of the switched system and consensus of the followers.
Future work will also extend the implementations to more realistic dynamic models for the followers and experiments on the hardware testbed.

\section{Acknowledgment}
This research is supported in part by AFRL award number FA9550-19-1-0169, DARPA award number D19AP00004, AFOSR award numbers FA9550-18-1-0109 and FA9550-19-1-0169, and NEEC award number N00174-18-1-0003. Any opinions, findings and conclusions or recommendations expressed in this material are those of the author(s) and do not necessarily reflect the views of the sponsoring agency.

\bibliographystyle{IEEEtran}
\bibliography{zherefclean_submit}

\end{document}